\documentclass[11pt]{amsart}

\usepackage{float}
\usepackage[centertags]{amsmath}
\usepackage{epsfig}
\usepackage{latexsym}
\usepackage{indentfirst}
\usepackage{mathrsfs}
\usepackage{amsthm}
\usepackage{amssymb}
\usepackage{amsfonts}
\usepackage{amsbsy}
\usepackage{appendix}
\usepackage[shortlabels]{enumitem}
\usepackage{url}
\usepackage{multirow}
\usepackage{array}
\usepackage{pdflscape}
\usepackage{xcolor}
\usepackage{caption} 
\usepackage{chngcntr}
\usepackage{stmaryrd}
\usepackage{verbatim}
\usepackage[normalem]{ulem}
\renewcommand{\arraystretch}{1.1}

\theoremstyle{plain}
\newtheorem{thm}{Theorem}[section]
\newtheorem{cor}[thm]{Corollary}
\newtheorem{lem}[thm]{Lemma}
\newtheorem{prop}[thm]{Proposition}

\theoremstyle{definition}
\newtheorem{defn}[thm]{Definition}
\newtheorem{rem}[thm]{Remark}
\newtheorem{ex}[thm]{Example}
\numberwithin{equation}{section}

\newcommand{\Fq}{\mathbb{F}_{q}}

\newcommand{\E}{\mathbb{E}}
\newcommand{\F}{{\mathbb F}}

\newcommand{\K}{\mathbb{K}}

\newcommand{\N}{\mathbb{N}}
\newcommand{\Tr}{{\rm Tr}}

\newcommand{\Span}{{\rm Span}}

\begin{document}

\title{Generalized Spectral Bound for Quasi-Twisted Codes}
\maketitle
\pagestyle{plain}

\author
{ {\scshape \begin{center} Buket \"{O}zkaya \end{center}}
\vspace{0.1cm} \small
\begin{center} Institute of Applied Mathematics,\\ Middle East Technical University, 06800 Ankara, Turkey\\
ozkayab@metu.edu.tr\\
\end{center}
}

\begin{abstract}
Semenov and Trifonov \cite{ST} developed a spectral theory for quasi-cyclic codes and formulated a BCH-like minimum distance bound. Their approach was generalized by Zeh and Ling \cite{LZ}, by using the HT bound.  The first spectral bound for quasi-twisted codes appeared in \cite{ELLOT}, which generalizes Semenov-Trifonov and Zeh-Ling bounds, but its overall performance was observed to be worse than the Jensen bound. More recently, an improved spectral bound for quasi-cyclic codes was proposed in \cite{LELO}, which outperforms the Jensen bound in many cases. In this paper, we adopt this approach to quasi-twisted case and we show that this new generalized spectral bound provides tighter lower bounds on the minimum distance compared to the Jensen and Ezerman et. al. bounds.
\end{abstract}


\section{Introduction}
\label{intro}
Thanks to their nice algebraic structure, several minimum distance bounds had been proposed for cyclic codes. Some of the well-known bounds are the Bose-Chaudhuri-Hocquenghem (BCH) bound \cite{BC,H}, the Hartmann-Tzeng (HT) bound \cite{HT} and the Roos bound \cite{R2}. All of these bound utilize the so-called zero set of a given cyclic code to estimate the minimum distance. On the other hand, constacyclic codes generalize cyclic codes in a natural way that shift invariance is accompanied with a multiplication by a nonzero constant. The aforementioned minimum distance bounds based on the zero sets have been generalized to constacyclic setup in \cite{RZ}.
 
Quasi-twisted (QT) codes form another important class of codes, which contains cyclic codes, quasi-cyclic (QC) codes and constacyclic codes as special subclasses. We refer the interested reader to \cite{Y,SZ} for an introduction to their algebraic structure. Furthermore, certain families of QT codes are also known to be asymptotically good \cite{C,DH,WS} and they yield good parameters \cite{AGOSS1,AGOSS2,QSS,SQS}.

Unfortunately, the study on the minimum distance of QT or QC codes is not adequate. Jensen \cite{J} and Lally \cite{L2} proved two bounds for QC codes, which hold for QT codes as well, and they use different aspects of the algebraic structure of QC codes. Semenov and Trifonov studied the spectral analysis of QC codes \cite{ST}, based on the work of Lally and Fitzpatrick in \cite{LF}, and formulated a BCH-like minimum distance bound. Their approach was generalized by Zeh and Ling in \cite{LZ}, by using the HT bound. The spectral analysis of QT codes appeared in \cite{ELLOT}, but the derived bound was outperformed by the Jensen bound. An improvement on this spectral bound has been provided in \cite{LELO} recently, which is valid for QC codes and beats the Jensen bound and the Ezerman et. al. bound.

In this paper, we investigate the spectral theory for QT codes further, by following the steps in \cite{ELLOT,LELO}. We derive a general spectral bound which contains the previous spectral bounds as special cases. For this, we provide a different proof than the ones in \cite{ELLOT,LELO}. Moreover, we push the general spectral bound to the largest possible extent and show that it holds for any subset of eigenvalues. Finally, the numerical examples and the simulations over random choices of QT codes are presented at the end.

This paper is organized as follows. Section \ref{basics} recalls some required background on the algebraic structure of constacyclic and QT codes. A generalized spectral bound on the minimum distance of QT codes is formulated and proven in Section \ref{bound sect}. Section \ref{res sect} contains explicit examples of QT codes and presents the performances of the generalized spectral bound, the Jensen bound and the old spectral bound in terms of sharpness and rank. All these computations were carried out using the computer algebra software \textsc{magma} \cite{BCP}.
 
\section{Background}\label{basics}

This section provides the necessary material on the minimum distance bounds for constacyclic and quasi-twisted codes, based on their algebraic properties. Let $\Fq$ denote the finite field with $q$ elements, where $q$ is a prime power, and let $\Fq^*=\Fq\setminus \{0\}$. For $n\geq 1$, a linear code $C\in \Fq^n$ is called an $[n,k,d]_q$-code if it is a $k$-dimensional subspace of $\Fq^n$ with minimum (Hamming) distance $d=d(C):=\min\{\mbox{wt}(\mathbf{c}) : 0\neq \mathbf{c}\in C\}$. We follow the notation and presentation in \cite{ELLOT}.

\subsection{Constacyclic codes} \label{consta sect}\hfill

Throughout the paper, let $m$ be a positive integer coprime with $q$. For a fixed element $\lambda \in \Fq^*$, a linear code $C\subseteq \Fq^m$ is called a $\lambda$-{\it constacyclic} code if it is invariant under the $\lambda$-constashift of codewords, that is, $(c_0,\ldots,c_{m-1}) \in C$ implies $(\lambda c_{m-1},c_0,\ldots,c_{m-2}) \in C$. In particular, if $\lambda = 1$ or $q=2$, then $C$ is a cyclic code. 

We define the quotient ring $R:=\Fq[x]/I$, where $I=\langle x^m -\lambda \rangle$. For an element $\mathbf{a}\in \Fq^m$, we associate an element of $R$ via the following isomorphism:
\begin{equation}\begin{array}{rcl} \label{identification-1}
\phi: \ \F_q^{m} & \longrightarrow & R  \\
\mathbf{a}=(a_0,\ldots,a_{m-1}) & \longmapsto & a(x):= a_0+\cdots + a_{m-1}x^{m-1}.
\end{array}
\end{equation}
Clearly, the $\lambda$-constashift in $\F_q^{m}$ corresponds to multiplication by $x$ in $R$. Therefore, a $\lambda$-constacyclic code $C\subseteq \Fq^m$ can be viewed as an ideal of $R$. Since $R$ is a principal ideal ring, there exists a unique monic $g(x)\in R$ such that $g(x)\mid x^m-\lambda$ and $C=\langle g(x)\rangle$. The polynomial $g(x)$ is called the {\it generator polynomial} of $C$, whereas $h(x)=\frac{x^m-\lambda}{g(x)}\in R$ is called the {\it check polynomial} of $C$. For $p\geq 1$, let $\mathbf{0}_p$ denote the all-zero vector of length $p$. We have $C=\{\mathbf{0}_m\}$ if and only if $g(x)=x^m-\lambda$. In this case, we assume throughout that $d(C)=\infty$.

Let $r$ be the smallest divisor of $q-1$ satisfying $\lambda^r=1$ and let $\alpha$ be a primitive $rm^{\rm th}$ root of unity with $\alpha^m=\lambda$. Then, $\xi:=\alpha^r$ is a primitive $m^{\rm th}$ root of unity and the roots of $x^m-\lambda$ are $\alpha, \alpha\xi, \ldots, \alpha\xi^{m-1}$. Henceforth, let $\Omega :=\{\alpha\xi^k : 0\leq k \leq m-1\}=\{\alpha^{1+kr} : 0\leq k \leq m-1\}$ be the set of all $m^{\rm th}$ roots of $\lambda$ and let $\F$ be the smallest extension of $\Fq$ that contains $\Omega$ (equivalently, $\F=\Fq(\alpha)$ so that $\F$ is the splitting field of $x^m-\lambda$). Given the $\lambda$-constacyclic code $C=\langle g(x)\rangle$, the set 
$$L:=\{\alpha\xi^k\ :\ g(\alpha\xi^k)=0\}\subseteq \Omega,$$ is called the {\it zero set} of $C$. The power set $\mathcal{P}(L)$ of $L$ is called the {\it defining set} of $C$. We have $L=\emptyset$ if and only if $C=\langle 1\rangle = \Fq^m$. Note that $\alpha\xi^k\in L$ implies $\alpha^q\xi^{qk}\in L$, for each $k$, where $\alpha^q\xi^{qk}=\alpha\xi^{k'}$ with $k'=\frac{q-1}{r}+qk \mod m$. A nonempty subset $E\subseteq\Omega$ is said to be {\it consecutive} if there exist integers $e,n$ and $\delta$ with $e\geq 0,\delta \geq 2, n> 0$ and $\gcd(m,n)=1$ such that
\begin{equation} \label{cons zero set}
E:=\{\alpha\xi^{e+zn}\ :\ 0\leq z\leq \delta-2\}\subseteq\Omega.
\end{equation}

Let $\mathcal{P}(\Omega)$ be the power set of $\Omega$. Note that any element $P\in\mathcal{P}(\Omega)$ is the zero set of some $\lambda$-constacyclic code $D_P\subseteq\F^m$ since $x^m-\lambda$ splits into linear factors over $\F$.
Let $C$ be a nontrivial $\lambda$-constacyclic code of length $m$ over some subfield of $\F$ with zero set $L\subseteq \Omega$. Then, for any $P\subseteq L$, $C$ is a subfield subcode of $D_P$ and therefore we have $d(C)\geq d(D_P)$. 
We define a {\it defining set bound} to be a member of a chosen family $\mathcal{B}(C):=\{(P, d_P)\} \subseteq \mathcal{P}(\Omega) \times (\N \cup \{\infty\})$ such that, for any $(P, d_P)\in \mathcal{B}(C)$, $P\subseteq L$ implies $d(C) \geq d(D_P) \geq d_P$. Note that any union of defining set bounds is again a defining set bound. We set 
\begin{equation*}\mathcal{B}_1(C) := \{(P, d(D_P)) : D_P\subseteq\F^m \mbox{\ has\ zero\ set\ } P, \mbox{\ for\ all\ } P\subseteq L\}.
\end{equation*}
In particular, when $P=L=\Omega$, we have $D_{\Omega}=\{\mathbf{0}_m\}$ over $\F$ with $d(D_{\Omega})=\infty$ and consequently, we include $(\Omega,\infty)$ in every collection $\mathcal{B}(C)$ as a convention when $L=\Omega$.

If we choose $$\mathcal{B}_2(C) := \{(E, |E|+1) : E\subseteq L \mbox{\ is\ consecutive}\},$$ then we obtain the BCH bound, where $E$ is of the form given in (\ref{cons zero set}). Similarly, we formulate the HT bound as \begin{equation*}\mathcal{B}_3(C) := \{(D, \delta + s) :
D = \{\alpha\xi^{e + z  n_1 + y  n_2} : 0\leq z \leq \delta-2, 0 \leq y\leq s\} \subseteq L\}, \end{equation*} for integers $ e \geq 0$, $\delta \geq 2$ and positive integers $s, n_1$ and $n_2$ such that $\gcd(m, n_1) = 1$ and $\gcd(m, n_2) < \delta$. 

We now present the Roos bound on the minimum distance of a given $\lambda$-constacyclic code (see \cite[Theorem 2]{R2} for the original version by C. Roos for cyclic codes). 

\begin{thm}\cite[Theorem 6]{RZ}. \label{Roos}
Let $N$ and $M$ be two nonempty subsets of $\Omega$. If there exists a consecutive set $M'$ containing $M$ such that $|M'| \leq |M| + d_N -2$, then we have $d_{MN}\geq |M| + d_N -1$ where $MN:=\displaystyle\frac{1}{\alpha}\bigcup_{\varepsilon\in M} \varepsilon N$.
\end{thm}
If $N$ is consecutive like in (\ref{cons zero set}), then we obtain the following.
\begin{cor}\cite[Corollary 1]{RZ},\cite[Corollary 1]{R2} \label{Roos2}
Let $N, M$ and $M'$ be as in Theorem \ref{Roos}, with $N$ consecutive. Then $|M'| < |M| + |N|$ implies $d_{MN}\geq |M| + |N|$.
\end{cor}

\begin{rem}\label{Roos remark}
The special case $M=\{\alpha\}$ gives the BCH bound for the associated constacyclic code (see \cite[Corollary 2]{RZ}). By taking $M'=M$, we obtain the HT bound (see \cite[Corollary 3]{RZ}). 
\end{rem}

Let $C$ be a nontrivial $\lambda$-constacyclic code of length $m$ over some subfield of $\F$ with zero set $L\subseteq \Omega$. Then, the Roos bound corresponds to the choice $$\mathcal{B}_4(C):=\{(MN,|M|+d_N-1) : \exists M'\subseteq\Omega\  \mbox{consecutive\ such\ that\ }  M'\supseteq M  \mbox{\ with\ } |M'|\leq |M|+d_N-2\},$$ for any $\emptyset\neq MN\subseteq L$ with $MN=\frac{1}{\alpha}\bigcup_{\varepsilon\in M} \varepsilon N$. 

\subsection{Quasi-twisted codes} \label{QT sect}\hfill

Let $\ell$ be a positive integer. A linear code $C\subseteq\Fq^{m\ell}$ is called a $\lambda$-quasi-twisted ($\lambda$-QT) code of index $\ell$ if it is invariant under the $\lambda$-constashift of codewords by $\ell$ positions and $\ell$ is the least positive integer satisfying this property. That is, $\mathbf{c}=(c_{00} ,\ldots, c_{0,\ell-1}, \ldots, c_{m-1,0}, \ldots, c_{m-1,\ell-1})\in C$ implies $(\lambda c_{m-1,0}, \ldots, \lambda c_{m-1,\ell-1}, c_{00} ,\ldots, c_{0,\ell-1}, \ldots, c_{m-2,0}, \ldots, c_{m-2,\ell-1})\in C$. In particular, if $\ell=1$, then $C$ is a $\lambda$-constacyclic code, and if $\lambda = 1$ or $q=2$, then $C$ is a QC code of index $\ell$. If we view a codeword $\mathbf{c}\in C$ as an $m \times \ell$ array
\begin{equation}\label{array}
\mathbf{c}=\left(
  \begin{array}{ccc}
    c_{00} & \ldots & c_{0,\ell-1} \\
    \vdots &  & \vdots \\
    c_{m-1,0} & \ldots & c_{m-1,\ell-1} \\
  \end{array}
\right),\end{equation} then being invariant under $\lambda$-constashift by $\ell$ positions in $\Fq^{m\ell}$ corresponds to being closed under row $\lambda$-constashift in $\Fq^{m\times\ell}$.

The isomorphism $\phi$ in (\ref{identification-1}) extends naturally to
\begin{equation}
\begin{array}{lll} \label{identification-2}
\Phi: \hspace{2cm} \F_q^{m\ell} & \longrightarrow & R^\ell  \\
\mathbf{c}=\left(
  \begin{array}{ccc}
    c_{00} & \ldots & c_{0,\ell-1} \\
    \vdots &  & \vdots \\
    c_{m-1,0} & \ldots & c_{m-1,\ell-1} \\
  \end{array}
\right) & \longmapsto & \hspace{-2pt}\mathbf{c}(x)\\
\end{array}\\
\end{equation}
such that to an element $\mathbf{c}\in \Fq^{m\times \ell} \simeq \Fq^{m\ell}$ represented as in (\ref{array}), we associate an element of $R^\ell$ 
\begin{equation} \label{associate-1}
\mathbf{c}(x):=(c_0(x),c_1(x),\ldots ,c_{\ell-1}(x)) \in R^\ell ,
\end{equation}
where, for each $0\leq j \leq \ell-1$, 
\begin{equation}\label{columns} 
c_j(x):= c_{0,j}+c_{1,j}x+c_{2,j}x^2+\cdots + c_{m-1,j}x^{m-1} \in R .
\end{equation} 

Observe that the row $\lambda$-constashift invariance in $\Fq^{m\times\ell}$ corresponds to being closed under componentwise multiplication by $x$ in $R^\ell$. Therefore, the map $\Phi$ above yields an $R$-module isomorphism and any $\lambda$-QT code $C\subseteq \F_q^{m\ell}\simeq \Fq^{m\times\ell}$ of index $\ell$ can be viewed as an $R$-submodule of $R^\ell$.

Any $\lambda$-QT code over $\F_q$ decomposes into shorter codes over extension fields of $\F_q$. For this description, we refer the reader to \cite{Y} for the respective proofs of the following assertions. Let $x^m-\lambda$ factor into irreducible polynomials in $\F_q[x]$ as
\begin{equation}\label{irreducibles}
x^m-\lambda=f_1(x)f_2(x)\cdots f_s(x).
\end{equation}
As $\gcd(m,q)=1$, there are no repeating factors in (\ref{irreducibles}). By the Chinese Remainder Theorem (CRT), we have the following ring isomorphism
\begin{equation} \label{CRT-1}
R\cong \bigoplus_{i=1}^{s} \F_q[x]/\langle f_i(x)\rangle .
\end{equation}
For each $i\in\{1,2,\ldots,s\}$, let $u_i$ be the smallest nonnegative integer such that $f_i(\alpha\xi^{u_i})=0$. Then, the direct summands in (\ref{CRT-1}) can be viewed as field extensions of $\F_q$, obtained by adjoining the element $\alpha\xi^{u_i}$. If we set $\E_i:=\Fq(\alpha\xi^{u_i})\cong \Fq[x]/\langle f_i(x) \rangle$, for each $1\leq i \leq s$, then $\E_i$ is an intermediate field between $\F$ and $\Fq$ such that $\big[\E_i : \F_q\big]=e_i$ and we have (cf. (\ref{CRT-1}))
\begin{eqnarray} \label{CRT-2}
R & \simeq & \E_1 \oplus \cdots \oplus \E_{s}
 \nonumber \\
a(x) & \mapsto & \left(a(\alpha\xi^{u_1}),\ldots ,a(\alpha\xi^{u_s}) \right).
\end{eqnarray}
This naturally extends to
\begin{eqnarray} \label{CRT-3}
R^{\ell}&\simeq & \E_1^{\ell} \oplus \cdots  \oplus \E_{s}^{\ell}
\nonumber \\ 
\mathbf{a}(x) & \mapsto & \left(\mathbf{a}(\alpha\xi^{u_1}),\ldots ,\mathbf{a}(\alpha\xi^{u_s}) \right),
\end{eqnarray}
where, for any $\mathbf{a}(x)=\bigl(a_{0}(x),\ldots ,a_{\ell-1}(x)\bigr)\in R^{\ell}$, $\mathbf{a}(\delta)$ denotes the componentwise evaluation at $\delta\in\F$. Hence, a $\lambda$-QT code $C\subseteq R^\ell$ can be viewed as an $(\E_1 \oplus \cdots \oplus \E_{s})$-submodule of $\E_1^{\ell} \oplus \cdots  \oplus \E_{s}^{\ell}$ and it decomposes as
\begin{equation} \label{constituents}
C \simeq C_1\oplus \cdots  \oplus C_{s},
\end{equation}
where $C_i$ is a linear code in $\E_i^{\ell}$, for each $i$. These linear codes over various extensions of $\F_q$ are called the {\it constituents} of $C$ (see \cite[\S 7]{Y} for explicit examples).

Let $C\subseteq R^\ell$ be generated as an $R$-module by
\[
\left\{\bigl(a_{1,0}(x),\ldots ,a_{1,\ell-1}(x)\bigr),\ldots,\bigl(a_{r,0}(x),\ldots ,a_{r,\ell-1}(x)\bigr)\right\}.
\]
Then, for $1\leq i \leq s$, we have
\begin{equation}\label{explicit constituents}
 C_i  = \Span_{\E_i}\bigl\{\bigl(a_{b,0}(\alpha\xi^{u_i}),\ldots, a_{b,\ell-1}(\alpha\xi^{u_i})\bigr): 1\leq b \leq r \bigr\}.   
\end{equation}

Equivalently, for $1\leq i \leq s$, each extension field $\E_i$ above is isomorphic to a minimal $\lambda$-constacyclic code of length $m$ over $\F_q$ with the irreducible check polynomial $f_i(x)$. If we denote by $\theta_i$ the generating primitive idempotent (see \cite[Theorem 1]{LG}) for the minimal $\lambda$-constacyclic code $\langle \theta_i \rangle$ in consideration, then the isomorphism is given by the maps
\begin{eqnarray} \label{isoms}\hspace{-10pt}
\begin{array}{ccl} \varphi_i:\langle \theta_i \rangle
& \longrightarrow & \E_i \\ \hspace{0.5cm} a(x)& \longmapsto &
a(\alpha\xi^{u_i}) \end{array}
&\hspace{-16pt} & \begin{array}{ccl} \psi_i: \E_i & \longrightarrow & \langle \theta_i \rangle \\
\hspace{0.5cm} \delta & \longmapsto & \sum\limits_{k=0}^{m-1} a_kx^k,
\end{array}
\end{eqnarray}
where
$$a_k=\frac{1}{m} \Tr_{\E_i/\F_q}(\delta\alpha^{-k}\xi^{-ku_i}).\vspace{8pt}$$
Note that, for each $1\leq i\leq s$, the maps $\varphi_i$ and $\psi_i$ are inverses of each other, regardless of the choice of the representative in the $\Fq$-conjugacy class $\big[\alpha\xi^{u_i}\big]$, since $\Tr_{\E_i/\F_q}(\epsilon^q)=\Tr_{\E_i/\F_q}(\epsilon)$, for any $\epsilon\in\E_i$.

If $\mathfrak{C}_i$ is a length $\ell$ linear code over $\E_i$, for $1\leq i \leq s$, then we denote its concatenation with $\langle \theta_i \rangle$ by $\langle \theta_i \rangle \Box \mathfrak{C}_i$, where this concatenation is carried out by the map $\psi_i$, extended to $\E_i^\ell$. In other words, $\psi_i$ is applied to each coordinate of the codewords in $\mathfrak{C}_i$ to produce an element of $\langle \theta_i\rangle ^\ell$. Here, $\langle \theta_i \rangle$ and $\mathfrak{C}_i$ are called the inner and outer codes of the concatenation, respectively. Jensen gave the following concatenated description for QT codes.

\begin{thm} \cite[Theorem 2]{LG} \label{Jensen's thm}\hfill
\begin{enumerate}
\item[i.] Let $C$ be a $\lambda$-QT code of length $m\ell$ and index $\ell$ over $\Fq$. Then, for some subset $\mathcal{I}$ of $\{1,\ldots ,s\}$, there exist linear codes $\mathfrak{C}_i$ of length $\ell$ over $\E_i$, which can be explicitly described, such that $$C=\displaystyle\bigoplus_{i\in \mathcal{I}} \langle \theta_i \rangle \Box \mathfrak{C}_i.$$
\item[ii.] Conversely, let $\mathfrak{C}_i$ be a linear code in $\E_i^\ell$, for each $i\in \mathcal{I} \subseteq \{1,\ldots ,s\}$. Then, $$C=\displaystyle\bigoplus_{i\in \mathcal{I}} \langle \theta_i \rangle \Box \mathfrak{C}_i$$ is a $q$-ary $\lambda$-QT code of length $m\ell$ and index $\ell$.
\end{enumerate}
\end{thm}

Moreover, each constituent $C_i$ in (\ref{constituents}) is equal to the outer code $\mathfrak{C}_i$ in the concatenated structure, for each $1\leq i \leq s$ (see \cite[Theorem 3]{LG}).

Jensen derived a minimum distance bound in \cite[Theorem 4]{J}, which is valid for all concatenated codes ({\it i.e.}, the inner and outer codes can be any linear code). Therefore, it applies to QT codes as well. We formulate the Jensen bound for QT codes as follows.
\begin{thm}  \label{Jensen bound}
Let $C\subseteq R^\ell$ be a $\lambda$-QT code with the concatenated structure $C=\bigoplus_{i\in \mathcal{I}} \langle \theta_i \rangle \Box C_i$, for some $\mathcal{I}\subseteq\{1,\ldots ,s\}$. Assume that $C_{i_1}, \ldots, C_{i_t}$ are the nonzero outer codes (constituents) of $C$, for $\{i_1,\ldots,i_t\} \subseteq \mathcal{I}$, such that $d(C_{i_1}) \leq d(C_{i_2})\leq \cdots \leq d(C_{i_t})$. Then, we have
\begin{equation}\label{Jensen}
d(C) \geq \displaystyle \min_{1\leq r \leq t} \left\{ d(C_{i_r}) d(\langle\theta_{i_1}\rangle \oplus \cdots \oplus \langle\theta_{i_r}\rangle) \right\}.
\end{equation}
\end{thm}
\vspace{7pt}

\subsection{Spectral theory for QT codes}    \label{spec sect} \hfill

We provide an adaptation of the results in \cite{LF} for QT codes, where it was shown that every QC code has a polynomial generating set in the form of a reduced Gr{\"o}bner basis. 

Consider the ring homomorphism:
\begin{eqnarray}\label{embedding}
\Psi \ :\ \Fq[x]^{\ell} &\longrightarrow& R^{\ell} \\\nonumber
(\widetilde{f}_0(x), \ldots ,\widetilde{f}_{\ell-1}(x))  &\longmapsto & 
(f_0(x), \ldots ,f_{\ell-1}(x)),
\end{eqnarray}
where $f_j(x)=\widetilde{f}_j(x) \mod \langle x^m-\lambda \rangle$, for $0\leq j \leq \ell-1$. Given a $\lambda$-QT code $C\subseteq R^{\ell}$, its preimage $\widetilde{C}$ in $\Fq[x]^{\ell}$ is an $\Fq[x]$-submodule containing $\widetilde{K} =\{(x^m-\lambda)\mathbf{e}_j : 0\leq j \leq \ell-1\}$, where each $\mathbf{e}_j$ denotes the standard basis vector of length $\ell$ with $1$ at the $j^{\rm th}$ coordinate and $0$ elsewhere. From this point on, the tilde will represent structures over $\Fq[x]$.

Since $\widetilde{C}$ is a submodule of the finitely generated free module $\F_q[x]^{\ell}$ over the principal ideal domain $\Fq[x]$ and contains $\widetilde{K}$, it has a generating set of the form $$\{\mathbf{u}_1,\ldots,\mathbf{u}_p, (x^m-\lambda)\mathbf{e}_0,\ldots,(x^m-\lambda)\mathbf{e}_{\ell-1}\},$$  where $p$ is a nonnegative integer and when $p>0$, $\mathbf{u}_b = (u_{b,0}(x),\ldots,u_{b,\ell-1}(x))\in \Fq[x]^{\ell}$, for each $b \in \{1,\ldots,p\}$. 
Hence, the rows of
$$\mathcal{G}=\left(\begin{array}{ccc}
    u_{1,0}(x) & \ldots & u_{1,\ell-1}(x) \\
    \vdots &  & \vdots \\
    u_{p,0}(x) & \ldots & u_{p,\ell-1}(x) \\
     x^m-\lambda & \ldots & 0 \\
    \vdots & \ddots & \vdots \\
    0 & \ldots & x^m-\lambda \\
  \end{array}
\right)$$
generate $\widetilde{C}$. By using elementary row operations, we obtain another equivalent generating set from the rows of an upper-triangular $\ell \times \ell$ matrix over $\Fq[x]$ as:
\begin{equation}\label{PGM}
\widetilde{G}(x)=\left(\begin{array}{cccc}
    g_{0,0}(x) & g_{0,1}(x) & \ldots & g_{0,\ell-1}(x) \\
    0 & g_{1,1}(x) & \ldots & g_{1,\ell-1}(x) \\
    \vdots & \vdots & \ddots & \vdots \\
    0 & 0 &\ldots & g_{\ell-1,\ell-1}(x)\\
  \end{array}
\right),
\end{equation}
where $\widetilde{G}(x)$ satisfies the following conditions (see \cite[Theorem 2.1]{LF}):
\begin{enumerate}
    \item $g_{i,j}(x)=0$, for all $0\leq j < i \leq \ell-1$.
    \item deg$(g_{i,j}(x)) < $ deg$(g_{j,j}(x))$, for all $i <j$.
    \item $g_{j,j}(x) \mid (x^m-\lambda)$, for all $0\leq j \leq \ell-1$.
    \item If  $g_{j,j}(x) = (x^m-\lambda)$, then $g_{i,j}(x) =0$, for all $i\neq j$.
\end{enumerate}
As the rows of $\widetilde{G}(x)$ are nonzero, each nonzero element of $\widetilde{C}$ can be expressed in the form \begin{center}$(0,\ldots,0,c_j(x),\ldots,c_{\ell-1}(x)),$ where $j\geq 0$, $ c_j(x)\neq 0$ and $g_{j,j}(x)\mid c_j(x)$. \end{center} Hence, the rows of $\widetilde{G}(x)$ form a Gr\"obner basis of $\widetilde{C}$ with respect to the position-over-term (POT) order in $\Fq[x]$, where the standard basis vectors $\{\mathbf{e}_0,\ldots,\mathbf{e}_{\ell-1}\}$ and the monomials $x^n$ are ordered naturally in each component. Moreover, the condition (2) above implies that the rows of $\widetilde{G}(x)$ provide a reduced Gr\"obner basis for $\widetilde{C}$, which is unique up to multiplication by constants, with monic diagonal elements.

Let $G(x)$ be the matrix with the rows of $\widetilde{G}(x)$ under the image of $\Psi$ in (\ref{embedding}). Clearly, the rows of $G(x)$ is an $R$-generating set for $C$. When $C$ is the trivial zero code of length $m\ell$, we have $p=0$ giving $G(x)=\mathbf{0}_{\ell}$. Otherwise, we say that $C$ is an $r$-generator $\lambda$-QT code, generated as an $R$-submodule, if $G(x)$ has $r$ (nonzero) rows. The $\Fq$-dimension of $C$ is given by (see \cite[Corollary 2.4]{LF} for the proof) 
\begin{equation}\label{dimension}
m\ell-\sum_{j=0}^{\ell-1}\mbox{deg}(g_{j,j}(x))=\sum_{j=0}^{\ell-1}\left[m -\mbox{deg}(g_{j,j}(x))\right].
\end{equation}

Given a $\lambda$-QT code $C\subseteq R^{\ell}$, let the associated $\ell \times \ell$ upper-triangular matrix $\widetilde{G}(x)$ be as in (\ref{PGM}) with entries in $\Fq[x]$. The {\it determinant} of $\widetilde{G}(x)$ is defined as $$\mbox{det}(\widetilde{G}(x)):=\prod_{j=0}^{\ell-1}g_{j,j}(x)$$ and an {\it eigenvalue} $\beta$ of $C$ is a root of det$(\widetilde{G}(x))$. Consequently, all eigenvalues are in $\Omega$ ({\it i.e.}, $\beta=\alpha\xi^k$, for some $k\in\{0,\ldots,m-1\}$), since $g_{j,j}(x)\mid (x^m-\lambda)$, for each $0\leq j\leq \ell-1$. The {\it algebraic multiplicity} of $\beta$ is the largest integer $a$ such that $(x-\beta)^a\mid\mbox{det}(\widetilde{G}(x))$. The {\it geometric multiplicity} of $\beta$ is defined as the dimension of the null space of $\widetilde{G}(\beta)$, where this null space is called the {\it eigenspace} of $\beta$ and it is denoted by $\mathcal{V}_{\beta}$. In other words, we have $$\mathcal{V}_{\beta}:=\{\mathbf{v}\in\F^{\ell} : \widetilde{G}(\beta)\mathbf{v}^{\top}=\mathbf{0}_{\ell}^{\top}\},$$ where $\F$ is the splitting field of $x^m-\lambda$ as before. Semenov and Trifonov \cite{ST} showed that, given a QC code and the associated $\widetilde{G}(x)$, the algebraic multiplicity $a$ of an eigenvalue $\beta$ is equal to its geometric multiplicity $\mbox{dim}_{\F}(\mathcal{V}_{\beta})$. The following QT analogue of this result can be shown in the same way.
\begin{lem}\cite[Lemma 1]{ST}\label{multiplicity lemma}
The algebraic multiplicity of any eigenvalue of a $\lambda$-QT code $C$ is equal to its geometric multiplicity.
\end{lem}

Henceforth, let $\overline{\Omega}\subseteq \Omega$ denote the set of all eigenvalues of $C$. Notice that $\overline{\Omega}=\emptyset$ if and only if the diagonal elements $g_{j,j}(x)$ in $\widetilde{G}(x)$ are constant polynomials and $C$ is the trivial full space code. For $|\overline{\Omega}|=t>0$, we choose an arbitrary eigenvalue $\beta_i\in\overline{\Omega}$ with multiplicity $n_i$, for some $i \in\{1,\ldots,t\}$. Let $\{\mathbf{v}_{i,0},\ldots,\mathbf{v}_{i,n_i-1}\}$ be a basis for the corresponding eigenspace $\mathcal{V}_i$. Consider the matrix  
\begin{equation}\label{Eigenspace} 
V_i:=\begin{pmatrix}
\mathbf{v}_{i,0} \\
\vdots\\
\mathbf{v}_{i,n_i-1} 
 \end{pmatrix}=\begin{pmatrix}
v_{i,0,0}&\ldots&v_{i,0,\ell-1} \\
\vdots &\vdots & \vdots\\
v_{i,n_i-1,0}&\ldots&v_{i,n_i-1,\ell-1}
 \end{pmatrix},
\end{equation} 
having the basis elements as its rows. We let
\[ H_i:=(1, \beta_i,\ldots,\beta_i^{m-1})\otimes V_i\ \]
and define
\begin{equation}\label{parity check matrix} 
H:=\begin{pmatrix}
H_1 \\
\vdots\\
H_t 
\end{pmatrix}=\begin{pmatrix}
V_1&\beta_1V_1&\ldots&\beta_1^{m-1}V_1 \\
\vdots &\vdots & & \vdots\\
V_t&\beta_tV_t&\ldots&\beta_t^{m-1}V_t
\end{pmatrix}.
\end{equation} 
Note that $H$ has $n:=\sum_{i=1}^t n_i$ rows. By Lemma \ref{multiplicity lemma}, we have $n=\sum_{j=0}^{\ell-1}\mbox{deg}(g_{j,j}(x))$. Together with the fact that all these $n$ rows are linearly independent, the following lemma can be shown easily.
\begin{lem}\cite[Lemma 2]{ST}\label{rank lemma}
The rank of the matrix $H$ in (\ref{parity check matrix}) is equal to $m\ell -\dim_{\Fq}(C)$.
\end{lem}

Observe that $H \mathbf{c}^{\top}=\mathbf{0}_n^{\top}$, for any codeword $\mathbf{c}\in C$. Together with Lemma \ref{rank lemma}, we easily obtain the following result (see \cite[Theorem 1]{ST} for the QC analogue of the result).

\begin{prop}
The $n\times m\ell$ matrix $H$ in (\ref{parity check matrix}) is a parity-check matrix for $C$.
\end{prop}

\begin{rem}\label{analogy rem}
The eigenvalues are the QT analogues of the zeros of constacyclic codes. Recall that a constacyclic code has an empty zero set if and only if it is equal to the full space. Similarly, $\overline{\Omega}=\emptyset$ if and only if $C=\Fq^{m\ell}$. In this case, the construction of the parity-check matrix $H$ in (\ref{parity check matrix}) is impossible, hence, we assume $H=\mathbf{0}_{m\ell}$. The other extreme case is when the zero set of a constacyclic code is $\Omega$, which implies that we have the trivial zero code. However, we emphasize that a $\lambda$-QT code with $\overline{\Omega}=\Omega$ is not necessarily the zero code. By Lemma \ref{rank lemma}, one can easily deduce that a given $\lambda$-QT code $C$ is the zero code $\{\mathbf{0}_{m\ell}\}$ if and only if $\overline{\Omega}=\Omega$ and each $\mathcal{V}_i=\F^{\ell}$ (equivalently, each $V_i=I_{\ell}$, where $I_{\ell}$ denotes the $\ell\times\ell$ identity matrix), and $n=m\ell$ so that we obtain $H=I_{m\ell}$. On the other hand, $\overline{\Omega}=\Omega$ whenever $(x^m-\lambda) \mid \mbox{det}(\widetilde{G}(x))$ but $C$ is nontrivial unless each $m^{\rm th}$ root of $\lambda$ in $\Omega$ has multiplicity $\ell$, which happens only if $\widetilde{G}(x)=(x^m-\lambda) I_{\ell}$.
\end{rem}

\begin{defn}\label{eigencode}
We define the {\it eigencode} corresponding to an eigenspace $\mathcal{V}\subseteq \F^\ell$ by
\[
\mathbb{C}(\mathcal{V})=\mathbb{C}:=\left\{\mathbf{u}\in \Fq^\ell\ : \ \sum_{j=0}^{\ell-1}{v_ju_j}=0, \forall \mathbf{v} \in \mathcal{V}\right\}.
\]
In case we have $\mathbb{C}=\{\mathbf{0}_{\ell}\}$, then it is assumed that $d(\mathbb{C})=\infty$.
\end{defn}

Semenov and Trifonov proved a BCH-like minimum distance bound for a given QC code (see \cite[Theorem 2]{ST}), which is expressed in terms of the size of a consecutive subset of eigenvalues in $\overline{\Omega}$ and the minimum distance of the common eigencode related to this consecutive subset. Zeh and Ling generalized their approach and derived an HT-like bound in \cite[Theorem 1]{LZ} without using the parity-check matrix in their proof. The eigencode, however, is still needed. In \cite{ELLOT}, a spectral bound based on any defining set bound is proven for any QT code with a nonempty set of eigenvalues. The QT analogues of Semenov-Trifonov and Zeh-Ling bounds were proven in terms of the Roos-like and shift-like bounds as a corollary. The next section revisits Ezerman et. al. bound with a different proof and extends this spectral bound to the general case of any set of eigenvalues. Then, we provide the QT analogue of Luo et. al. bound \cite{LELO}.

\section{Generalized spectral bound for quasi-twisted codes}\label{bound sect}

The spectral bound proven in \cite[Theorem 13]{ELLOT} was shown for QT codes with nonempty eigenvalue set $\overline{\Omega}\subseteq\Omega$. In fact, the bound remains valid when $\overline{\Omega}=\emptyset$. To restate and prove the bound, we need to fix some notation first.

Let $A$ be an $k\times n$ matrix with entries from some finite field $\K$. Let 
\begin{equation}\label{nA}
n_A:=\max\{j : 0\leq j \leq n \mbox{\ and\ any\ } j \mbox{\ columns\ of\ } A \mbox{\ are\ linearly\ independent}  \}.
\end{equation}
It is well known that, given a linear code $C$ over $\K$ with parity check matrix $H$, we have $d(C)=n_H+1$ (e.g. see \cite[Corollary 4.5.7 ]{LX}). 

We need the following lemma to restate and prove the spectral bound in \cite{ELLOT}.
\begin{lem}\label{bound lemma}
Let $A$ and $B$ be two matrices with entries from some finite field $\K$.
\begin{itemize}
\item[i.] If $A$ and $B$ have the same number of columns, then the matrix $M={A\choose B}$ satisfies $n_M=\max\{n_A,n_B\}$.
\item[ii.] If we set $M=(A\otimes B)$, then we have $n_M=\min\{n_A,n_B\}$.
\end{itemize}
\end{lem}

\begin{proof}\hfill
\begin{itemize}
\item[i.] Clearly, the columns of the matrix $M$ are of the form $$(M_1 \cdots M_n)=\begin{pmatrix}
A_{1} & \cdots & A_{n}\\
B_{1} & \cdots & B_{n}
\end{pmatrix},$$
where $A_{j}$ and $B_j$ denote the $j$'th column of the matrix $A$ and $B$, respectively, for $j\in\{1,\ldots,n\}$.

Without loss of generality, let $n_A\geq n_B$. We first choose $x_1, x_2, \ldots, x_{n_B}\in\K$ such that $x_1 B_{j_1} + \cdots + x_{n_B} B_{j_{n_B}}=0$, where $j_1,\ldots,j_{n_B}\in\{1,\ldots,n\}$ are arbitrary. Then, we need to choose $n_A-n_B$ more columns in $A$ different from $A_{j_1},\ldots,A_{j_{n_B}}$, together with some elements $x_{n_B+1}, \ldots, x_{n_A}\in\K$ so that
$$x_1 B_{j_1} + \cdots + x_{n_B} B_{j_{n_B}} + x_{n_B+1} A_{j_{n_B+1}} \cdots + x_{n_A} A_{j_{n_A}}=0.$$
This implies $x_1=\cdots = x_{n_A}=0$, otherwise it would contradict the definition of $n_A$ and $n_B$ given in \eqref{nA} above. Therefore, $x_1 M_1 + \cdots + x_{n_A} M_{n_A}=0$ implies $x_1=\cdots = x_{n_A}=0$ and we have $n_M=n_A$ since $j_1,\ldots,j_{n_B}$ are chosen arbitrarily.

\item[ii.] We start with the case when $n_A\geq n_B$. Recall that the Kronecker product satisfies
\[M=(A\otimes B)=\begin{pmatrix}
a_{11}B & \cdots & a_{1n}B\\
\vdots & & \vdots\\
a_{m1}B & \cdots & a_{mn}B
\end{pmatrix}.\]
We again choose $x_1, x_2, \ldots, x_{n_B}\in\K$ such that $x_1 B_{j_1} + \cdots + x_{n_B} B_{j_{n_B}}=0$, for arbitrary $j_1,\ldots,j_{n_B}\in\{1,\ldots,n\}$. This implies $x_1 a_{it} B_{j_1} + \cdots + x_{n_B} a_{it} B_{j_{n_B}}=0$, which holds for any $i\in\{1,\ldots,m\}$ and $t\in\{1,\ldots,n\}$. Hence, we can extend the above sum over all $i$ and $t$, obtaining $n_M=n_B$ as $x_1=\cdots = x_{n_B}=0$ by the definition of $n_B$.

We now consider the case when $n_A\leq n_B$. Let $y_1, y_2, \ldots, y_{n_A}\in\K$ such that $y_1 A_{t_1} + \cdots + y_{n_A} A_{t_{n_A}}=0$, for arbitrary $t_1,\ldots,t_{n_A}\in\{1,\ldots,n\}$. In particular, for a fixed $i\in\{1,\ldots,m\}$, we have $y_1 a_{it_1} + \cdots + y_{n_A} a_{it_{n_A}}=0$ and this equality is preserved if we multiply each summand with $B$. Hence, we obtain $n_M=n_A$.
\end{itemize}
\end{proof}

\begin{thm}\cite[Theorem 13]{ELLOT} \label{main thm new}
Given a $\lambda$-QT code $C$ of index $\ell$ with eigenvalue set $\overline{\Omega}\subseteq\Omega$, let $D_{\overline{\Omega}}$ be the $\lambda$-constacyclic code of length $m$ over $\F$ with zero set $\overline{\Omega}$ and let $\mathcal{B}(D_{\overline{\Omega}}) \subseteq \mathcal{P}(\Omega) \times (\N \cup \{\infty\})$ be an arbitrary family of defining set bounds for $D_{\overline{\Omega}}$. For any $P\subseteq \overline{\Omega}$ such that $(P,d_P)\in \mathcal{B}(D_{\overline{\Omega}})$, we define $\mathcal{V}_{P}:=\bigcap_{\beta\in P}\mathcal{V}_{\beta}$ as the common eigenspace of the eigenvalues in $P$ and let $\mathbb{C}_{P}=(\mathcal{V}_{P})^{\perp} \bigr\rvert_{\Fq}$ denote the corresponding eigencode. Then, 
\[
d(C) \geq \min \left\{ d_P, d(\mathbb{C}_P) \right\}.
\]
\end{thm}

\begin{proof}
Given the $\lambda$-QT code $C$ with eigenvalue set $\overline{\Omega}\subseteq\Omega$, we first consider the case when $\overline{\Omega}=\emptyset$ (i.e. $C=\Fq^{m\ell}$). As the construction of a parity check matrix as in \eqref{parity check matrix} is not possible, we set $V_{\emptyset}=\mathbf{0}_\ell$ and $H_{\emptyset}=\mathbf{0}_m$, which gives $H=H_{\emptyset}\otimes V_{\emptyset}=\mathbf{0}_{m\ell}$. Recall that $D_{\emptyset}=\F^m$ and $d(D_{\emptyset})=1$. If we include $(\emptyset,d_{\emptyset})=(\emptyset,1)$ into $\mathcal{B}(D_{\overline{\Omega}})$ and use the fact that $\mathbb{C}_{\emptyset}=(\mathcal{V}_{\emptyset})^{\perp} \bigr\rvert_{\Fq}=\Fq^\ell$, then we obtain $d(C)\geq \min \left\{ d_{\emptyset}, d(\mathbb{C}_{\emptyset}) \right\}=\min\{1,1\}=1$.

Now we continue with the other extreme case when $P=\overline{\Omega}=\Omega$ and $d_{\Omega}=d(D_{\Omega})$, where the $\lambda$-constacyclic code $D_{\Omega}$ has $\Omega$ as its zero set and we always have $d_{\Omega}=\infty$ in this case since $D_{\Omega}=\{\mathbf{0}_{m}\}$. Moreover, by Lemma 12 in \cite{ELLOT}, $C$ is the zero code if and only if $\mathbb{C}_{\Omega}=\{\mathbf{0}_{\ell}\}$ with $d(\mathbb{C}_{\Omega})=\infty$. Hence, both $d_{\Omega}$ and $d(\mathbb{C}_{\Omega})$ become $\infty$ if and only if $C=\{\mathbf{0}_{m\ell}\}$ and we get $d(C) \geq \min \left\{ d_{\Omega}, d(\mathbb{C}_{\Omega}) \right\}=\infty$.

Thus, it remains to show that $d(C) \geq \min \left\{d_P, d(\mathbb{C}_P) \right\}$, for any fixed $\emptyset \neq P\subseteq \overline{\Omega} \subseteq \Omega$ such that $(P, d_P)\in \mathcal{B}(D_{\overline{\Omega}})$ and $d_P$ is finite. For this, we assume that $P =\{\alpha\xi^{u_1}, \alpha\xi^{u_2},\ldots,\alpha\xi^{u_r}\}\subseteq \overline{\Omega}$, where $0<r< m$.  We define 
\begin{equation}\label{pmatrix}
\widetilde{H}_P:=\begin{pmatrix}
1&\alpha\xi^{u_1}&(\alpha\xi^{u_1})^2&\ldots&(\alpha\xi^{u_1})^{m-1}\\
\vdots & \vdots & \vdots & \vdots & \vdots \\
1&\alpha\xi^{u_r}&(\alpha\xi^{u_r})^2&\ldots&(\alpha\xi^{u_r})^{m-1}
\end{pmatrix}.
\end{equation}
Recall that $P$ is the zero set of some $D_P\subseteq\F^m$, which contains $D_{\overline{\Omega}}$ as a subcode, and $\widetilde{H}_P$ is a parity-check matrix of this $D_P$. Note that $d(D_{\overline{\Omega}})\geq d(D_P) \geq d_P$, by definition.

We have assumed $P \neq \emptyset$ to make sure that $\widetilde{H}_P$ is well-defined. In particular, when $\mathcal{V}_P=\{\mathbf{0}_{\ell}\}$, which implies $\mathbb{C}_P= \Fq^\ell$ and $d(\mathbb{C}_P) =1$, we have $d(C) \ge 1 = \min\{d_P,d(\mathbb{C}_P)\}$ since $d_P\geq 1$. If $\mathcal{V}_P\neq\{\mathbf{0}_{\ell}\}$, then let $V_P$ be the matrix, say of size $t\times\ell$, whose rows form a basis for the common eigenspace $\mathcal{V}_P$ (cf. (\ref{Eigenspace})). If we set $\widehat{H}_P :=\widetilde{H}_P \otimes V_P$, then $\widehat{H}_P \, \mathbf{c}^{\top}=\mathbf{0}_{rt}^{\top}$, for all $\mathbf{c}\in C$. In other words, $\widehat{H}_P$ is a submatrix of some matrix $H$ of the form in (\ref{parity check matrix}). The code $\widehat{C}$ over $\F_q$ with parity-check matrix $\widehat{H}_P$ contains the code $C$, which implies that $d(C)\geq d(\widehat{C})$. The proof follows by Lemma \ref{bound lemma} and the observations $d_P\leq d(D_P)=n_{\widetilde{H}_P}+1$ and $d(\mathbb{C}_P) =n_{V_P}+1$.
\end{proof}

Theorem \ref{main thm new} above allows us to use \emph{any} defining set bound derived for constacyclic codes. The following special case is immediate after the preparation in Section \ref{basics} (cf. Theorem \ref{Roos}, Corollary \ref{Roos2} and Remark \ref{Roos remark}). 

\begin{cor}\label{Cor-Roos new}
Let $C\subseteq R^\ell$ be a $\lambda$-QT code of index $\ell$ with $\overline{\Omega}\subseteq\Omega$ as its nonempty set of eigenvalues. Let $N$ and $M$ be two nonempty subsets of $\Omega$ such that $MN\subseteq\overline{\Omega}$, where $MN:=\frac{1}{\alpha}\bigcup_{\varepsilon\in M} \varepsilon N$. If there exists a consecutive set $M'\supseteq M$ with $|M'|\leq |M|+d_N-2$, then $d(C)\geq \min(|M|+d_N-1,d(\mathbb{C}_{MN}))$.
\end{cor}

\begin{rem}\label{ext remark}
By using Remark \ref{Roos remark}, we can obtain the QT analogues of the BCH-like bound given in \cite[Theorem 2]{ST} and the HT-like bound in \cite[Theorem 1]{LZ}. 
\end{rem}

Let the $\lambda$-QT code $C$ with eigenvalue set $\overline{\Omega}\subseteq\Omega$ and the associated $\lambda$-constacyclic code $D_{\overline{\Omega}}\subseteq \F^m$ with the selected collection of defining set bounds $\mathcal{B}(D_{\overline{\Omega}})$ be given as in Theorem \ref{main thm new}. From this point on, we denote the estimate of the spectral bound by $$d_{Spec}(\mathcal{B}(D_{\overline{\Omega}}) ; P,d_P):=\min\{d_P, d(\mathbb{C}_P)\},$$ where $\emptyset\neq P \subseteq \overline{\Omega}$ such that $(P,d_P)\in\mathcal{B}(D_{\overline{\Omega}})$, and we set 
\[
d_{Spec}(\mathcal{B}(D_{\overline{\Omega}})):=\displaystyle\max_{\substack{(P, d_P)\in \mathcal{B}(D_{\overline{\Omega}}) \\  P \subseteq \overline{\Omega}}}\left\{d_{Spec}(\mathcal{B}(D_{\overline{\Omega}}) ; P,d_P)\right\}.
\]

\begin{rem}
Example 1 in \cite{ELLOT} shows how the spectral bound beats the Jensen bound. Taking a sufficiently large $P$ as well as considering different choices of defining set bounds is crucial when using the spectral bound. For instance, for any fixed choice of defining set bound $\mathcal{B}(D_{\overline{\Omega}})$, if we consider the behaviour of $d_{Spec}(\mathcal{B}(D_{\overline{\Omega}}) ; P, d_P)$ over the nonempty subsets $P\subseteq \overline{\Omega}$ with $(P,d_P)\in\mathcal{B}(D_{\overline{\Omega}})$, then $d_P$ might be nondecreasing as $P$ approaches $\overline{\Omega}$, whereas $d(\mathbb{C}_P)$ is nonincreasing. Hence, the optimized value $d_{Spec}(\mathcal{B}(D_{\overline{\Omega}}))$ defined above might not be obtained by involving all eigenvalues, unlike the case of constacyclic codes, as any nonconstant $d_P$ grows proportionately to the size of $P$. 
\end{rem}

Theorem \ref{main thm new} shows that each defining set bound for constacyclic codes is also applicable to QT codes. We now provide a generalized spectral bound on the minimum distance of QT codes, which allows us to use more than one defining set bound on the given eigenvalue set. This new bound outperforms the above spectral bound as well as the Jensen bound. First, we need the following proposition, which generalizes Lemma \ref{bound lemma}.

\begin{prop}\label{bound prop}
For a positive integer $s$, let $A_1,\ldots, A_s$ be $k_1 \times n_1$ matrices and let $B_1,\ldots, B_s$ be $k_2 \times n_2$ matrices over some finite field $\K$ such that $n_{A_1}\geq n_{A_2}\geq \cdots \geq n_{A_s}$. We define \[M:=\begin{pmatrix}
A_1 \otimes B_1\\
\vdots\\
A_s \otimes B_s
\end{pmatrix}.\]
Then, we have $n_M\geq \min\{n_{A_1},n_{A_2}n_{B_1}, n_{A_3}n_{B_{1+2}},\ldots, n_{A_s}n_{B_{1+2+\cdots +(s-1)}},n_{B_{1+2+\cdots +s}}\},$ where \[B_{1+2+\cdots +j}=\begin{pmatrix}
 B_1\\
 B_2\\
\vdots\\
 B_j
\end{pmatrix},\]
for any $j\in\{1,\ldots,s\}$.
\end{prop}
\begin{proof}
The proof follows by induction on $s$. Notice that the case $s=1$ corresponds to the second part of Lemma \ref{bound lemma}. Now let $s=2$ and $M:=\begin{pmatrix}
A_1 \otimes B_1\\
A_2 \otimes B_2
\end{pmatrix}$ with $n_{A_1}\geq n_{A_2}$. In the cases we will consider below, we use Lemma \ref{bound lemma} and the following fact: if $n_{A_1}\geq n_{B_1}$ and $n_{B_2}\geq n_{A_2}$, then we have $n_M= n_{A_2} n_{B_1}$. To see this, we can take a nonzero vector $c\in\K^{n_1n_2}$ and consider $Mc^\top = 0$. To satisfy this homogeneous equality, all of the $n_1$ blocks of length $n_2$ in $c$ must have weight at least $n_{B_1}+1$ or be zero so that $B_1(c_1)^\top=\cdots=B_1(c_{n_1})^\top=0$ and at least $n_{A_2}+1$ out of these $n_1$ blocks must be nonzero (hence, of weight at least $n_{B_1}+1$) so that $A_2$ is cancelled out. 

\emph{Case 1: When $n_{A_1}\geq n_{A_2}\geq n_{B_{1+2}}$,} the number of linearly independent columns are determined by $B_1$ and $B_2$, which follows by Lemma \ref{bound lemma}. Hence, we get $n_M\geq n_{B_{1+2}}=\max\{n_{B_1},n_{B_2}\}$.

\emph{Case 2: When $n_{B_{1+2}}\geq n_{A_1}\geq n_{A_2}$,} the number of linearly independent columns are determined by $A_1$ and $A_2$, which follows by Lemma \ref{bound lemma}. Here, $n_M\geq\max\{n_{A_1},n_{A_2}\}=n_{A_1}$.

\emph{Case 3: When $n_{A_1}\geq n_{B_{1+2}}\geq n_{A_2}$,} we have to distinguish two possibilities. First, we consider when $n_{B_{1+2}}=n_{B_1}$. This means we either have $n_{A_1}\geq n_{B_1}\geq n_{B_2}\geq n_{A_2}$, implying $n_M= n_{A_2} n_{B_1}$, or $n_{A_1}\geq n_{B_1}\geq n_{A_2}\geq n_{B_2}$, implying $n_M=n_{B_{1+2}}$. Second, when $n_{B_{1+2}}=n_{B_2}$, then we have $n_M= n_{A_2} n_{B_1}$.

Therefore, we obtain $n_M\geq\min\{n_{A_1}, n_{A_2} n_{B_1}, n_{B_{1+2}}\}$ for $s=2$.

Assume that the assertion holds for $s-1$, i.e., 
\begin{equation}\label{ind-hyp}
n_M\geq \min\{n_{A_1},n_{A_2}n_{B_1}, \ldots, n_{A_{s-1}}n_{B_{1+2+\cdots +(s-2)}},n_{B_{1+2+\cdots +(s-1)}}\},
\mbox{\ where\ } M:=\begin{pmatrix}
A_1 \otimes B_1\\
\vdots\\
A_{s-1} \otimes B_{s-1}
\end{pmatrix}.
\end{equation}

By adding $A_s\otimes B_s$ to $M$, we want to show the proposition. Mimicing the case analysis done for $s=2$ above, we consider the following:

\emph{Case 1: When $n_{A_1}\geq \cdots\geq n_{A_S}\geq n_{B_{1+\cdots +s}}$,} the number of linearly independent columns rely on $B_1, \ldots, B_s$, which again comes from Lemma \ref{bound lemma}. Hence, we get $n_M\geq n_{B_{1+\cdots + s}}$.

\emph{Case 2: When $n_{B_{1+\cdots +s}}\geq n_{A_1}\geq \cdots\geq n_{A_s}$,} then we have to consider $A_1,\ldots,A_s$, which again follows by Lemma \ref{bound lemma}. Here, $n_M\geq\max\{n_{A_1},n_{A_2},\ldots, n_{A_s}\}=n_{A_1}$.

\emph{Case 3: When $n_{A_1}\geq \cdots \geq n_{A_{s-1}}\geq n_{B_{1+\cdots +s}}\geq n_{A_s}$,} we have to distinguish two possibilities, as done above. First, we consider when $n_{B_{1+\cdots +s}}=n_{B_{1+\cdots +(s-1)}}$ giving $n_M\geq n_{A_s} n_{B_{1+\cdots +(s-1)}}$ or  $n_M\geq n_{B_{1+\cdots +s}}$. Second, when $n_{B_{1+\cdots +s}}=n_{B_s}$, then we have $n_M\geq n_{A_s} n_{B_{1+\cdots +(s-1)}}$.

The remaining cases correspond to
\begin{itemize}
\item $n_{A_1}\geq \cdots \geq n_{A_{s-2}}\geq n_{B_{1+\cdots +s}}\geq n_{A_{s-1}}\geq n_{A_s}$,\\
 $\vdots$
\item $ n_{A_1}\geq n_{B_{1+\cdots +s}}\geq n_{A_2}\geq \cdots \geq n_{A_s}$,
\end{itemize}
which are covered by the induction hypothesis \eqref{ind-hyp}. More precisely, when $n_{A_1}\geq \cdots \geq n_{A_{s-j-1}}\geq n_{B_{1+\cdots +s}}\geq n_{A_{s-j}}\geq \cdots \geq n_{A_s}$, for some $1\leq j \leq s-2$, we again distinguish two scenarios. If $n_{B_{1+\cdots +s}}=n_{B_{1+\cdots +(s-j-1)}}$, then $n_M\geq n_{A_{s-j}} n_{B_{1+\cdots +(s-j-1)}}$ or  $n_M\geq n_{B_{1+\cdots +s}}$. Otherwise, if $n_{B_{1+\cdots +s}}=n_{B_{s-k}}$, for some $k\in\{0,\ldots, j\}$, then we have $n_M\geq n_{A_{s-j}} n_{B_{1+\cdots +(s-j-1)}}$.
\end{proof}

\begin{thm}\label{ISB}
Let $C$ be a $\lambda$-QT code of length $m\ell$ and index $\ell$ with eigenvalue set $\overline{\Omega}$. Let $D_{\overline{\Omega}}$ be the $\lambda$-constacyclic code of length $m$ over $\F$ with zero set $\overline{\Omega}$ and let
\[
(P_1,d_{P_1}),\ldots,(P_s,d_{P_s})\in \mathcal{B}(D_{\overline{\Omega}})
\]
be some defining set bounds, where the subsets $P_1,\ldots,P_s$ are not necessarily disjoint. Without loss of generality, we assume that $d_{P_1}\geq \cdots \geq d_{P_s}$. If $\mathbb{C}_{P_i}$ is the eigencode of $\mathcal{V}_{P_i}=\bigcap_{\beta\in P_i}\mathcal{V}_{\beta}$, for each $1\leq i\leq s$, then
\begin{equation}\label{NewSpec}
d\geq \min \left\{d_{P_1},d_{P_2} \, d(\mathbb{C}_{P_1}), d_{P_3} \, d\left(\bigcap_{i=1}^2\mathbb{C}_{P_i}\right), \cdots, d_{P_s} \, d\left(\bigcap_{i=1}^{s-1}\mathbb{C}_{P_i}\right), d\left(\bigcap_{i=1}^s\mathbb{C}_{P_i}\right)\right\}.
\end{equation}
\end{thm}

\begin{proof}
Note that, when $s=1$, we have the setup of Theorem \ref{main thm new}, which includes the cases of the trivial zero code and the full space code as well. Therefore, we can assume that the eigenvalue set $\overline{\Omega}$ is nonempty and $s\geq 2$.
For each $1\leq i\leq s$, let $P_i=\{\beta_1^{(i)},\cdots,\beta_{r_i}^{(i)}\}$ and let $V_{P_i}=\left(\mathbf{v}_1^{(i)},\cdots,\mathbf{v}_{n_i}^{(i)}\right)^{\top}$ be an $n_i\times \ell$ matrix whose rows are formed by a basis of $\mathcal{V}_{P_i}=\bigcap_{\beta\in P_i}\mathcal{V}_{\beta}$. For each $1\leq i\leq s$, we define the $r_i \times m$ matrix
\[
\widetilde{H}_{P_i}:=\begin{pmatrix}
1&  \beta_1^{(i)} & \cdots& \left(\beta_1^{(i)}\right)^{m-1}\\
\vdots& \vdots &  \ddots & \vdots\\
1&  \beta_{r_i}^{(i)} & \cdots& \left(\beta_{r_i}^{(i)}\right)^{m-1}
\end{pmatrix}
\]
We can verify that $\widetilde{H}_{P_i}$ is a parity-check matrix of some $\lambda$-constacyclic code $D_{P_i}$ with zero set $P_i$. The minimum distance of $D_{P_i}$ is greater than or equal to $d_{P_i}$. Then, the matrix
\[
\widehat{H}=\begin{pmatrix}
\widetilde{H}_{P_1}\otimes V_{P_1}\\
\vdots\\
\widetilde{H}_{P_s}\otimes V_{P_s}
\end{pmatrix}
\]
is a submatrix of a parity-check matrix of $C$ defined in \eqref{parity check matrix}. The code $\widehat{C}$ over $\F_q$ with parity-check matrix $\widehat{H}$ contains the code $C$, which implies that $d(C)\geq d(\widehat{C})$. The proof follows by Proposition \ref{bound prop} and the fact that $\left(\mathcal{V}_{P_1}+\cdots+\mathcal{V}_{P_j}\right)^\perp\bigr\rvert_{\Fq}=\bigcap_{i=1}^j\mathbb{C}_{P_i}$, for any $j\in\{1,\ldots,s\}$ (see Section 3 in \cite{ELLOT}).
\end{proof}

Let $C$ be a $\lambda$-QT code of length $m\ell$ and index $\ell$ with eigenvalue set $\overline{\Omega}$. Let$D_{\overline{\Omega}}$ be the $\lambda$-constacyclic code of length $m$ over $\F$ with zero set $\overline{\Omega}$ and let
\[
(P_1,d_{P_1}),\ldots,(P_s,d_{P_s})\in \mathcal{B}(D_{\overline{\Omega}})
\]
be some defining set bounds for $D_{\overline{\Omega}}$. Henceforth, we denote the estimate of the generalized spectral bound by
\[
d_{Spec}(\mathcal{B}(D_{\overline{\Omega}}); P_1,\ldots,P_s) := \min\left\{d_{P_1}, d_{P_2} \, d(\mathbb{C}_{P_1}), \cdots, d_{P_s} \, d\left(\bigcap_{i=1}^{s-1}\mathbb{C}_{P_i}\right), d\left(\bigcap_{i=1}^s\mathbb{C}_{P_i}\right)\right\},
\]
and we set
\begin{equation}\label{optbound}
d_{Spec}(\mathcal{B}(D_{\overline{\Omega}}),s):=\max_{\substack{(P_i,d_{P_i})\in \mathcal{B}(D_{\overline{\Omega}}) \\ 1\leq i \leq s}} \{d_{Spec}(\mathcal{B}(D_{\overline{\Omega}}); P_1,\ldots,P_s)\}.
\end{equation}
The generalized spectral bound $d_{Spec}(\mathcal{B}(D_{\overline{\Omega}}); P_1,\ldots,P_s)$ considers $s$ out of $|\mathcal{B}(D_{\overline{\Omega}})|$ defining set bounds and the optimized value $d_{Spec}(\mathcal{B}(D_{\overline{\Omega}}),s)$ maximizes over all $ \binom{|\mathcal{B}(D_{\overline{\Omega}})|}{s}$ possible outcomes.

\begin{rem}\label{remark1}
Theorem \ref{ISB} generalizes and improves Theorem \ref{main thm new} as the latter follows when $s=1$ in Theorem \ref{ISB}. Whenever $s>1$, since $d\left(\bigcap_{i=1}^s \mathbb{C}_{P_i}\right) \geq d\left(\bigcap_{i=1}^{s-1} \mathbb{C}_{P_i}\right) \geq \cdots \geq d(\mathbb{C}_{P_1} \cap \mathbb{C}_{P_2}) \geq d(\mathbb{C}_{P_1})$ and each $d_{P_i}\geq 1$, we obtain
\[
\min\left\{d_{P_1}, d_{P_2} \, d(\mathbb{C}_{P_1}),\cdots, d_{P_s} \, d\left(\bigcap_{i=1}^{s-1}\mathbb{C}_{P_i}\right), d \left(\bigcap_{i=1}^s \mathbb{C}_{P_i} \right)\right\}\geq \min\{d_{P_1}, d(\mathbb{C}_{P_1})\}.
\]

Assume that $(P_i,d_{P_i})=(P_j,d_{P_j})$, for some $1\leq i<j\leq s$. Without loss of generality, we can rearrange the list and assume that $j=i+1$. Since $d_{P_i}=d_{P_{i+1}}$ and $\mathbb{C}_{P_i}=\mathbb{C}_{P_{i+1}}$, we have $d_{P_i} \, d\left(\bigcap_{k=1}^{i}\mathbb{C}_{P_k}\right)=d_{P_{i+1}} \, d\left(\bigcap_{k=1}^{i+1}\mathbb{C}_{P_k}\right)$. Hence, considering an identical tuples $(P_i,d_{P_i})=(P_j,d_{P_j})$ among the $s$ terms is equivalent to considering the bound given in Theorem \ref{ISB} with $s-1$ terms. Therefore, the bound with $s$ terms, optimized as in \eqref{optbound}, already includes the respective estimates produced with $s-1$, $s-2$, $\cdots$, $2$ terms and a single term. We note that, in this repetition-free consideration, $s \leq |\mathcal{B}(D_{\overline{\Omega}})|$.
\end{rem}

\section{Examples and comparison results} \label{res sect}

First, we would like to emphasize that the spectral bound given in Theorem \ref{main thm new} cannot exceed $\max\{n,\ell\}$. If the given $\lambda$-QT code $C$ has the eigenvalue set $\overline{\Omega}=\Omega$ and we set $P_1=\Omega$ and the eigencodes satisfy $\bigcap_{i=1}^s\mathbb{C}_{P_i}=\{\mathbf{0}\}$, then the generalized spectral bound in Theorem \ref{ISB} might become greater than $\max\{n,\ell\}$, as the following example shows.

\begin{ex}\label{example1}
The ternary $[8,2,6]$ QT code with $\lambda=2$ and $(m,\ell)=(4,2)$, which is generated by 
\[
(a_{1,0}(x),a_{1,2}(x))=(2x^3 + 2x + 1, x^3 + x + 2),
\]
 is optimal (see www.codetables.de). Its upper-triangular generator matrix is
\[
\widetilde{G}(x) =
\begin{pmatrix}
x^2 + 2x + 2 & 2x^2 + x + 1 \\
0 & x^4+1
\end{pmatrix},
\]
with $\det(\widetilde{\mathbf{G}}(x))=(x^2 + 2x + 2)(x^4+1)$ and, hence, $\overline{\Omega}=\Omega$. Over $\F_3$, the scalar generator matrix is
\[
\setcounter{MaxMatrixCols}{20}
\begin{pmatrix}
1 & 0 & 1 & 1 & 2 & 0 & 2 & 2 \\
0 & 1 & 1 & 2 & 0 & 2 & 2 & 1
\end{pmatrix}.
\]
The generalized spectral bound sharply yields $d_{Spec}(\mathcal{B}(D_{\overline{\Omega}}),2)=6$, whereas, with $n=4$ and $\ell=2$, the spectral bound in Theorem \ref{main thm new} cannot exceed $4$. In fact, it yields $3$ as the estimate.
\end{ex}

In \cite{ELLOT}, the Jensen bound was shown to outperform the spectral bound in Theorem \ref{main thm new}, based on the simulation results. That spectral bound, in turn, is sharper than the Lally bound (see \cite{L2}). We now consider the performance of the generalized spectral bound in Theorem \ref{ISB} against the Jensen bound (denoted by $d_J$) and the spectral bound (denoted by $d_S$) in the QT setup. 

Given a $\lambda$-QT code $C$ of index $\ell$ with eigenvalue set $\overline{\Omega}\subseteq\Omega$, in the example below we take a subset $\widehat{\mathcal{B}}(D_{\overline{\Omega}})$ of the defining set bounds $\mathcal{B}(D_{\overline{\Omega}})$ consisting of four types of bounds, namely the BCH, HT, and Roos bounds, as well as $(P,d(D_{P})) \in D(\mathfrak{C})$, where $P$ is a subset of $\overline{\Omega}$ and $D_P$ is the corresponding $\lambda$-constacyclic code with zero set $P$, which contains $D_{\overline{\Omega}}$ as a subcode. In the search routines, we use the \texttt{MAGMA} functions for the Roos, HT, and BCH bounds from \cite{Piva2014}. The restricted estimate of the generalized spectral bound $d_{Spec}(\widehat{\mathcal{B}}(D_{\overline{\Omega}}),s)$ is calculated for $s=2$.

\begin{ex}\label{example2}
The $[16,8,3]_3$ QT code with $\lambda=2$ and $(m,\ell)=(4,4)$ generated by
\begin{align*}
(a_{1,0}(x), \ldots, a_{1,3}(x)) &=(x^3 + x^2 + 1, 2x^2 + x + 1, x^3 + 2x^2 + x, 2x^3 + 2x^2 + x + 2)\\
\mbox{ and } \\
(a_{2,0}(x), \ldots, a_{2,3}(x)) &=(x^3 + 2x + 1, 2x^3 + x^2 + x + 2, x^2 + x + 1, 2x^3 + 2x^2 + 2x)
\end{align*}
has an upper-triangular generator matrix
\[
\widetilde{G}(x) =
\begin{pmatrix}
1 &  2x^3 + 2x^2 + 2x + 2 &  0 &  2x^3 + 2x^2 + x\\
0 &  x^4 + 1 &  0 &  0\\
0 &  0 &  1 &  x^3 + 2x^2 + 2x + 2\\
[0 &  0 &  0 &  x^4 + 1
\end{pmatrix}.
\]
Since $\det(\widetilde{G}(x)) =(x^4+1)^2$, we infer that $\overline{\Omega}=\Omega$. Over $\F_3$, the scalar generator matrix is
\[
\setcounter{MaxMatrixCols}{20}
\begin{pmatrix}
1 & 0 & 0 & 0 & 2 & 2 & 2 & 2 & 0 & 0 & 0 & 0 & 0 & 1 & 2 & 2\\
0 & 1 & 0 & 0 & 1 & 2 & 2 & 2 & 0 & 0 & 0 & 0 & 1 & 0 & 1 & 2\\
0 & 0 & 1 & 0 & 1 & 1 & 2 & 2 & 0 & 0 & 0 & 0 & 1 & 1 & 0 & 1\\
0 & 0 & 0 & 1 & 1 & 1 & 1 & 2 & 0 & 0 & 0 & 0 & 2 & 1 & 1 & 0\\
0 & 0 & 0 & 0 & 0 & 0 & 0 & 0 & 1 & 0 & 0 & 0 & 2 & 2 & 2 & 1\\
0 & 0 & 0 & 0 & 0 & 0 & 0 & 0 & 0 & 1 & 0 & 0 & 2 & 2 & 2 & 2\\
0 & 0 & 0 & 0 & 0 & 0 & 0 & 0 & 0 & 0 & 1 & 0 & 1 & 2 & 2 & 2\\
0 & 0 & 0 & 0 & 0 & 0 & 0 & 0 & 0 & 0 & 0 & 1 & 1 & 1 & 2 & 2\\
\end{pmatrix}. \vspace{10pt}
\]
The generalized spectral bound in (\ref{optbound}) yields $d_{Spec}(\mathcal{B}(D_{\overline{\Omega}}),2)=3$, whereas the Jensen bound in (\ref{Jensen}) and the spectral bound in Theorem \ref{main thm new} provide the same estimate $d_J=d_S=2$. 
\end{ex}

Given a $\lambda$-QT code $C$ of index $\ell$ with eigenvalue set $\overline{\Omega}\subseteq\Omega$, the computation of the optimized generalized spectral bound with $s$ terms in \eqref{optbound} requires to consider $(2^{|\overline{\Omega}|}-1)^s$ possible $s$-tuples of nonempty subsets $P_1, \cdots, P_s \subseteq \overline{\Omega}$. Moreover, if $k$ types of defining set bounds are considered in $\mathcal{B}(D_{\overline{\Omega}})$, then $k^s$ possible $s$-tuples of defining set bounds have to regarded per choice of $P_1, \cdots, P_s \subseteq \overline{\Omega}$. In total, there are $k^s \left(2^{|\overline{\Omega}|}-1\right)^s$ comparisons to be made for each randomly generated QT code. In Table \ref{table:simulation} and in the examples above, we have therefore taken $k=4$ and choose $s \in \{2,3,4\}$ to make the computations practically possible. Overall, in order to make the bound computation more efficient than the brute force minimum weight search, the aforementioned number of comparisons should be kept much smaller than the number of nonzero codeword representatives $\dfrac{q^{\text{dim}(C)}-1}{q-1}=\dfrac{q^{n\ell - |\overline{\Omega}|}-1}{q-1}$.

Table \ref{table:simulation} presents the overall performances of the generalized spectral bound, the Jensen bound and the spectral bound in the ternary case with $\lambda=2$. Namely, we construct random ternary $2$-QT codes on each input tuple $(n,\ell,r)$, with $n = 4$ fixed, to keep the size of the eigenvalue set $|\overline{\Omega}|$ manageable, and use the ranges $2 \leq \ell \leq 4$ and $1 \leq r \leq \ell$. Once $(n,\ell,r)$ was determined, an array $\mathcal{A}$ of $r \ell$ generator polynomials was randomly built. The {\tt QuasiTwistedCyclicCode} function of \texttt{MAGMA} on input $(n \ell, \mathcal{A}, r)$ produced the corresponding $r$-generator quasi-cyclic code $C$. Then we computed its minimum distance $d(C)$ and the estimates given by the three bounds $d_{Spec}, d_J$ and $d_{S}$, where $d_{Spec} :=d_{Spec}(\widehat{\mathcal{B}}(D_{\overline{\Omega}}),s)$, with $s=2,3$ or $4$ and $\widehat{\mathcal{B}}(D_{\overline{\Omega}})$ consisted of the $k=4$ types of defining set bounds as mentioned earlier for Example \ref{example2}. We listed the respective numbers of occasions when each bound was either sharp or best-performing, with double counting allowed.

\begin{table*}[ht!]
\renewcommand{\arraystretch}{1.2}
\caption{The performance comparison of the minimum distance bounds in the ternary $2$-QT case. We list the number of nontrivial instances, when a specified bound reached the actual minimum distance and when it was greater than or equal to the other bound.}
\centering
\begin{tabular}{c|c|c|c|c|c}
	\hline
	$q=3$ &  $d_{Spec}, s=2$ &  $d_{Spec}, s=3$ &  $d_{Spec}, s=4$ & $d_J$ & $d_{S}$    \\
	\hline
	sharp & 42 & 45 & 45 & 35 & 32  \\
	best-performing & 112 & 120 & 120 & 108 & 73 \\
	\hline
	  \multicolumn{6}{c}{ \# nontrivial $C$ : 135}  \\
	\hline
\end{tabular}
\label{table:simulation}
\end{table*}

\section*{Acknowledgement}
The author would like to thank Markus Grassl for his help to run the simulations presented in Table \ref{table:simulation} above. The author is supported by T\"UB{\.I}TAK project 223N065.

\end{document}